\documentclass{article}
\usepackage[utf8]{inputenc}

\usepackage[margin=1in]{geometry}
\usepackage{authblk}
\usepackage{microtype}
\usepackage{hyperref}
\usepackage{graphicx}
\usepackage[margin=5pt,
font+=small,labelformat=parens,labelsep=space,
skip=6pt,list=false,hypcap=false
]{subcaption}

\usepackage[numbers,sort&compress]{natbib}

\usepackage{amsfonts}
\usepackage{amsmath}
\usepackage{amsthm}
\usepackage{textcomp}

\newtheorem{theorem}{Theorem}
\newtheorem{lemma}{Lemma}
\newtheorem{definition}{Definition}

\newtheorem{observation}{Observation}
\newtheorem{problem}{Problem}

\title{Subquadratic Encodings for Point Configurations}

\author[1]{Jean Cardinal\thanks{%
\texttt{jcardin@ulb.ac.be}.
Supported by the ``Action de Recherche Concert\'ee'' (ARC)
COPHYMA, convention number 4.110.H.000023.}}
\author[2]{Timothy M. Chan\thanks{\texttt{tmc@illinois.edu}.}}
\author[1]{John Iacono\thanks{%
\texttt{arxiv@johniacono.com}.
Supported by NSF grants CCF-1319648, CCF-1533564, CCF-0430849
and MRI-1229185, a Fulbright Fellowship and by the Fonds de la Recherche
Scientifique-FNRS under Grant n\textdegree{} MISU F 6001 1 and two Missions
Scientifiques.}}
\author[1]{\\Stefan Langerman\thanks{%
\texttt{slanger@ulb.ac.be}.
Directeur de recherches du Fonds de la Recherche Scientifique-FNRS.}}
\author[1]{Aurélien Ooms\thanks{%
\texttt{aureooms@ulb.ac.be}.
Supported by the Fund for Research Training in Industry and Agriculture (FRIA).}}

\affil[1]{Département d'Informatique, Université libre de Bruxelles (ULB), Belgium}
\affil[2]{Department of Computer Science, University of Illinois at Urbana-Champaign, USA}

\date{}

\begin{document}
\maketitle

\begin{abstract}
	For most algorithms dealing with sets of points in the plane, the only
	relevant information carried by the input is the combinatorial
	configuration of the points: the orientation of each triple of points in
	the set (clockwise, counterclockwise, or collinear). This information is
	called the \emph{order type} of the point set.
	In the dual, realizable order types and abstract order types are
	combinatorial analogues of line arrangements and pseudoline arrangements.
	Too often in the literature we analyze algorithms in the
	real-RAM model for simplicity, putting aside the fact that computers as we
	know them cannot handle arbitrary real numbers without some sort of
	encoding.
	Encoding an order type by the integer coordinates of some realizing point
	set is known to yield doubly exponential coordinates in some cases. Other
	known encodings can achieve quadratic space or fast orientation queries,
	but not both.
	In this contribution, we give a compact encoding for abstract order types
	that allows efficient query of the orientation of any triple: the encoding
	uses \( O(n^2) \) bits and an orientation query takes \(O(\log n)\) time in
	the word-RAM model.
	This encoding is space-optimal for abstract order types. We show how to
	shorten the encoding to \(O(n^2 {(\log\log n)}^2 / \log n)\) bits for
	realizable order types, giving the first subquadratic encoding for those
	order types with fast orientation queries.
	We further refine our encoding to attain \(O(\log n/\log\log n)\)
	query time without blowing up the space requirement.
	In the realizable case, we show that all those encodings can be computed
	efficiently.
	Finally, we generalize our results to the encoding of point configurations
	in higher dimension.
\end{abstract}
\section{Introduction}

At SoCG'86, Chazelle asked~\cite{GP93}:
\begin{quotation}
``How many bits does it take to know an order type?''
\end{quotation}

This question is of importance in Computational Geometry for the following two
reasons:
First, in many algorithms dealing with sets of points in the plane,
the only relevant information carried by the input is the combinatorial
configuration of the points given by the orientation of each triple of points in the
set (clockwise, counterclockwise, or collinear)~\cite{Ed12}.
Second, computers as we know them can only handle numbers with
finite description and we cannot assume that they can handle arbitrary
real numbers without some sort of encoding. The study of \emph{robust}
algorithms is focused on ensuring the correct solution of problems on finite
precision machines. Chapter 41 of The Handbook of Discrete and Computational
Geometry is dedicated to this issue~\cite{Ya04}.

The (counterclockwise) orientation \(\nabla(p,q,r)\) of a triple of points
\(p\), \(q\), and \(r\) with coordinates \((x_p, y_p)\), \((x_q, y_q)\), and
\((x_r, y_r)\) is the sign of the determinant
\begin{displaymath}
\begin{vmatrix}
	1 & x_p & y_p \\
	1 & x_q & y_q \\
	1 & x_r & y_r
\end{vmatrix}.
\end{displaymath}

Given a set of \(n\) labeled points \(P = \{\, p_1, p_2, \ldots, p_n\,\}\), we
define the \emph{order type} of \(P\) to be the function \(\chi \colon\,
{[n]}^3 \to \{\, -, 0, +\,\} \colon\, (a,b,c) \mapsto \nabla(p_a, p_b, p_c)\)
that maps each triple of point labels to the orientation of the corresponding
points, up to isomorphism.
A great deal of the literature in computational geometry deals with this
notion~\cite{%
Le26,
Ri56,
FL78,
Go80,
GP83,
GP84,
GP86,
Al86,
GPS89,
Ri89,
BRS92,
GP91,
GP93,
BLSWZ93,
Fe96,
NV98,
AK01,
BMS01,
AAK02a,
AAK02b,
RZ04,
AK05,
FV11,
HMMS11,
MMIB12,
AMP13,
AILOW14,
AKPV14,
AKMPW15,
ACKLV16%
}.
The order type of a point set has been further abstracted into combinatorial
objects known as (rank-three) \emph{oriented matroids}~\cite{FL78}. The
\emph{chirotope axioms} define consistent systems of signs of
triples~\cite{BLSWZ93}.
From the topological representation theorem~\cite{BMS01}, all such
\emph{abstract} order types correspond to pseudoline arrangements, while, from
the standard projective duality, order types of point sets correspond to
straight line arrangements. See Chapter 6 of The Handbook for more
details~\cite{RZ04}.

When the order type of a pseudoline arrangement can be realized by an
arrangement of straight lines, we call the pseudoline arrangement
\emph{stretchable}.
As an example of a nonstretchable arrangement, Levi gives Pappus's
configuration where eight triples of concurrent straight lines force a ninth,
whereas the ninth triple cannot be enforced by pseudolines~\cite{Le26} (see
Figure~\ref{fig:pappus}).
Ringel shows how to convert the so-called ``non-Pappus'' arrangement of
Figure~\ref{fig:pappus}~(b) to a simple arrangement while preserving
nonstretchability~\cite{Ri56}.
All arrangements of eight or fewer pseudolines are stretchable~\cite{GP80}, and
the only nonstretchable simple arrangement of nine pseudolines is the one given
by Ringel~\cite{Ri89}.
More information on pseudoline arrangements is available in Chapter 5 of The
Handbook~\cite{Go04}.

\begin{figure}
	\centering{}
\begin{subfigure}[t]{0.5\textwidth}
		\centering{}
		\includegraphics[scale=.9]{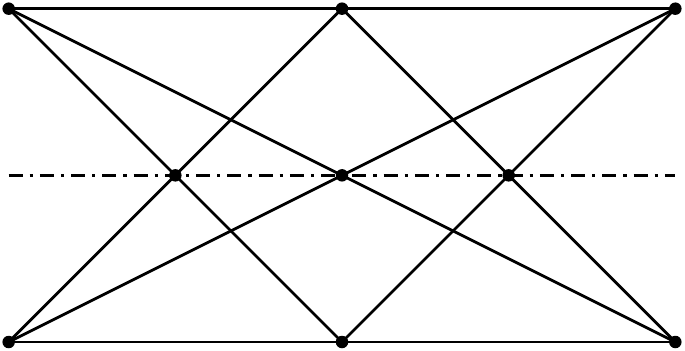}
		\caption{Realizable order type.}
\end{subfigure}%
\begin{subfigure}[t]{0.5\textwidth}
		\centering{}
		\includegraphics[scale=.9]{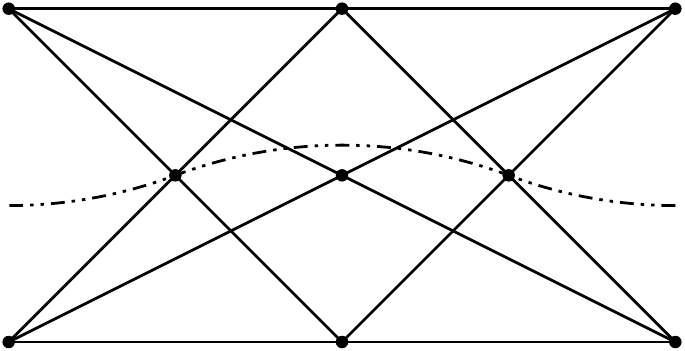}
		\caption{Abstract order type which is not realizable.}
\end{subfigure}
	\caption{Pappus's configuration.}\label{fig:pappus}
\end{figure}

Figure~\ref{fig:pappus} shows that not all pseudoline arrangements are
stretchable. Indeed, most are not: there are \(2^{\Theta(n^2)}\)
abstract order types~\cite{Fe96} and only \(2^{\Theta(n \log n)}\) realizable
order types~\cite{GP86,Al86}.
This discrepancy stems from the algebraic nature of realizable order types, as
illustrated by the main tool used in the upper bound proofs (the Milnor-Thom
Theorem~\cite{Mi64,Th65}).

Information theory implies that we need quadratic space for abstract order
types whereas we only need linearithmic space for realizable order types.
Hence, storing all \( \binom{n}{3} \) orientations in a lookup table seems
wasteful.
Another obvious idea for storing the order type of a point set is to store
the coordinates of the points, and answer orientation queries
by computing the corresponding determinant. While this should work in many practical
settings, it cannot work for all point sets. Perles's configuration shows that
some configuration of points, containing collinear triples, forces at least one
coordinate to be irrational~\cite{Gr05} (see Figure~\ref{fig:perles}).
Order types of points in general position can always be represented by rational
coordinates. It is well known, however, that some configurations require doubly
exponential coordinates, hence coordinates with exponential bitsizes if
represented in the normal way~\cite{GPS89}.

\begin{figure}
	\centering{}
	\includegraphics[scale=.9]{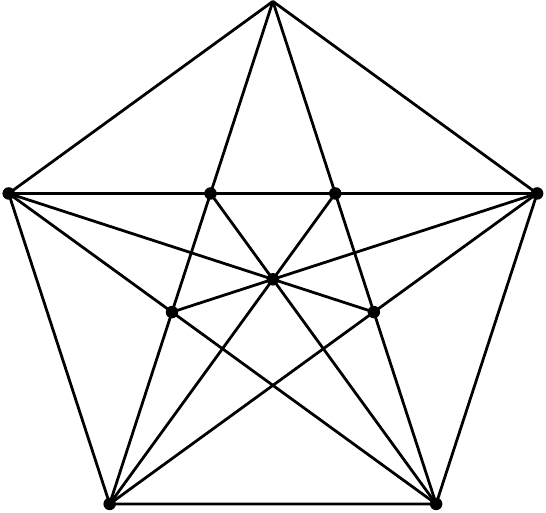}
	\caption{Perles's configuration.}\label{fig:perles}
\end{figure}

Goodman and Pollack defined \(\lambda\)-matrices which can encode abstract order
types using \( O(n^2 \log{n}) \) bits~\cite{GP83}. They asked if the space
requirements could be moved closer to the information-theoretic lower bounds.
Felsner and Valtr showed how to encode abstract order types optimally in
\(O(n^2)\) bits via the wiring diagram of their corresponding allowable
sequence~\cite{Fe96, FV11} (as defined in~\cite{Go80}). Aloupis et al.\ gave
an encoding of size \(O(n^2)\) that can be computed in \(O(n^2)\) time and that
can be used to test for the isomorphism of two distinct point sets in the same
amount of time. However, it is not known how to decode the orientation of one
triple from any of those encodings in, say, sublinear time. Moreover, since the
information-theoretic lower bound for realizable order types is only \(\Omega(n
\log{n})\), we must ask if this space bound is approachable for those order
types while keeping orientation queries reasonably efficient.
\subsection*{Our Results}

In this contribution, we are interested in \emph{compact} encodings for
order types: we wish to design data structures using as few bits as possible
that can be used to quickly answer orientation queries of a given abstract or
realizable order type.
In Section~\ref{sec:abstract}, we give the first optimal encoding for abstract
order types that allows efficient query of the orientation of any triple: the
encoding is a data structure that uses \( O(n^2) \) bits of space with queries
taking \(O(\log n)\) time in the word-RAM model.
Our encoding is far from being space-optimal for realizable order types.
Fortunately, its construction can be easily tuned to only require \(O(n^2
{(\log{\log{n}})}^2 / \log{n})\) bits in this case.
In Section~\ref{sec:query}, we further refine our encoding to
reduce the query time to \(O(\log{n}/\log{\log{n}})\).
In the realizable case, we give quadratic upper bounds on the
preprocessing time required to compute an encoding in the real-RAM model.
In Section~\ref{sec:hyperplanes} we generalize our encodings for chirotopes of
point sets in higher dimension.

Our data structure is the first subquadratic encoding for realizable order
types that allows efficient query of the orientation of any triple. It is not
known whether a subquadratic constant-degree algebraic decision tree exists for
the related problem of deciding whether a point set contains a collinear
triple. Any such decision tree would yield another subquadratic encoding for
realizable order types. We see the design of compact encodings for realizable
order types as a subgoal towards subquadratic nonuniform algorithms for this
related problem, a major open problem in Computational Geometry. Note that
pushing the preprocessing time below quadratic would yield such an algorithm.
\section{Encoding Order Types via Hierarchical Cuttings}\label{sec:abstract}

To make our statements clear, we will use the following definition:
\begin{definition}\label{def:encoding}
For fixed \(k\) and given a function \(f : {[n]}^k \to [O(1)]\), we define
a \((S(n),Q(n))\)-encoding of \(f\) to be a string of \(S(n)\) bits such
that, given this string and any \(t \in {[n]}^k\), we can compute \(f(t)\)
in \(Q(n)\) query time in the word-RAM model.
\end{definition}
In this section, we use this definition with \(f\) being some order
type,\footnote{%
Technically, we encode the orientation predicate of some realizing
arrangement of the order type and skip the isomorphism. If desired, a
canonical labeling of the arrangement can be produced in \(O(n^2)\) time for
abstract and realizable order types~\cite{AILOW14}.
}
\(k=3\) and the codomain of \(f\) being \(\{\, -,0,+\,\}\). For the rest of the
discussion, we assume the word-RAM model with word-size \(w \geq \log n\) and
the standard arithmetic and bitwise operators.
We prove our main theorems for the two-dimensional case:
\begin{theorem}\label{thm:abstract}
All abstract order types have a \((O(n^2), O(\log n))\)-encoding.
\end{theorem}
\begin{theorem}\label{thm:realizable}
All realizable order types have a
\((O(\frac{n^2 {(\log \log n)}^2}{\log n}), O(\log n))\)-encoding.
\end{theorem}
\begin{theorem}\label{thm:preprocessing}
In the real-RAM model and the constant-degree algebraic decision tree model,
given \(n\) real-coordinate input points in \(\mathbb{R}^2\) we can compute
the encoding of their order type as in
Theorems~\ref{thm:abstract}~and~\ref{thm:realizable} in \(O(n^2) \) time.
\end{theorem}
For instance, Theorem~\ref{thm:realizable} implies that for any set of points
\(\{\, p_1, p_2, \ldots, p_n\,\}\), there exists a string of \(O(n^2 {(\log
\log n)}^2 / \log n)\) bits such that given this string and any triple of
indices \((a,b,c) \in {[n]}^3\) we can compute the value of \(\chi(a,b,c) =
\nabla(p_a, p_b, p_c)\) in \(O(\log n)\) time.

Throughout the rest of this paper, we will assume that we can access some
arrangement of lines or pseudolines that realizes the order type we want to
encode. We thus exclusively focus on the problem of encoding the order type of
a given arrangement. This does not pose a threat against the existence of an
encoding. However, we have to be more careful when we bound the preprocessing
time required to compute such an encoding. This is why, in
Theorem~\ref{thm:preprocessing}, we specify the model of computation and how
the input is given.

\paragraph*{Hierarchical Cuttings} We encode the order type of an arrangement via
hierarchical cuttings as defined in~\cite{C93}. A cutting in \( \mathbb{R}^d \)
is a set of (possibly unbounded and/or non-full dimensional)
constant-complexity cells that together partition \(\mathbb{R}^{d}\). A
\(\frac{1}{r}\)-cutting of a set of \(n\) hyperplanes is a cutting with the
constraint that each of its cells is intersected by at most \(\frac{n}{r}\)
hyperplanes. There exist various ways of constructing \(\frac{1}{r}\)-cuttings of
size \(O(r^d)\). Those cuttings allow for efficient divide-and-conquer
solutions to many geometric problems. The hierarchical cuttings of Chazelle
have the additional property that they can be composed without multiplying the
hidden constant factors in the big-oh notation. In particular, they allow
for \(O(n^d)\)-space \(O(\log n)\)-query \(d\)-dimensional point location data
structures (for constant \(d\)). In the plane, hierarchical cuttings can be
constructed for arrangement of pseudolines with the same properties.

\paragraph*{Idea} We want to preprocess \(n\) pseudolines \( \{\, \ell_1 ,
\ell_2 , \ldots, \ell_n\,\} \) in the plane so that, given three indices \(a\),
\(b\), and \(c\), we can compute their orientation, that is, whether the
intersection \(\ell_a \cap \ell_b\) lies above, below or on \(\ell_c\). Our
data structure builds on cuttings as follows: Given a cutting \(\Xi\) and the
three indices, we can locate the intersection of \(\ell_a\) and \(\ell_b\)
inside \(\Xi\). The location of this intersection is a cell of \(\Xi\). The
next step is to decide whether \(\ell_c\) lies above, lies below, contains or
intersects that cell. In the first three cases, we are done. Otherwise, we
can answer the query by recursing on the subset of pseudolines intersecting the
cell containing the intersection. We build on hierarchical cuttings to solve
all subproblems efficiently.

\paragraph*{Intersection Location} When the \(\ell_a\) are straight lines,
locating the intersection \(\ell_a \cap \ell_b\) in \(\Xi\) is trivial if we
know the real parameters of \(\ell_a\) and \(\ell_b\) and of the descriptions
of the subcells of \(\Xi\). However, in our model we are not allowed to store
real numbers. To circumvent this annoyance, and to handle arrangements of
pseudolines, we make a simple observation illustrated by
Figure~\ref{fig:permutation}.
\begin{figure}
\centering{}
\begin{subfigure}[t]{0.5\textwidth}
\centering{}
\includegraphics[scale=.9]{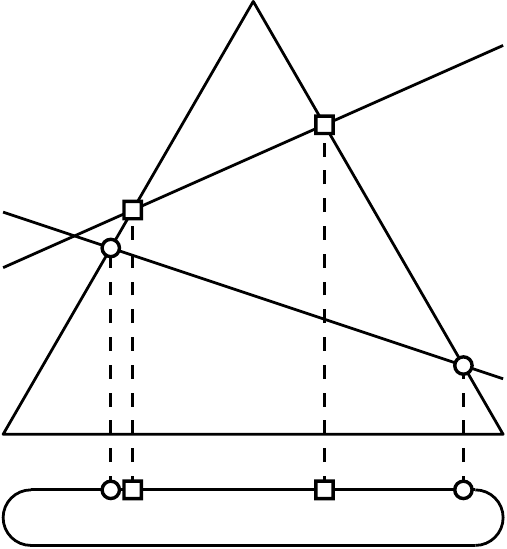}
\caption{\(\ell_a \cap \ell_b \cap \mathcal{C} = \emptyset\).}
\end{subfigure}%
\begin{subfigure}[t]{0.5\textwidth}
\centering{}
\includegraphics[scale=.9]{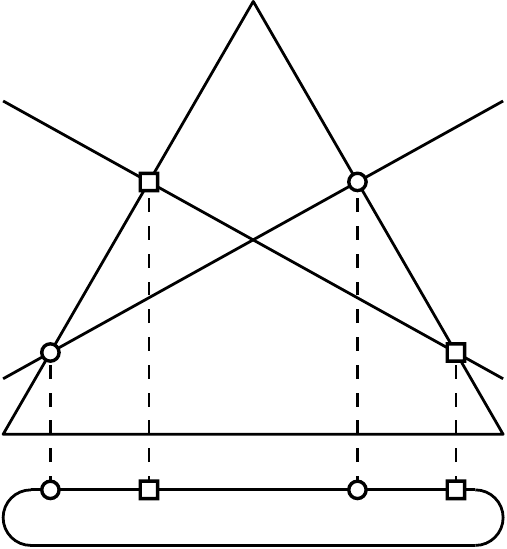}
\caption{\(\ell_a \cap \ell_b \cap \mathcal{C} \neq \emptyset\).}
\end{subfigure}
\caption{Cyclic permutations.}\label{fig:permutation}
\end{figure}
\begin{observation}
Two pseudolines \(\ell_a\) and \(\ell_b\) intersect in the interior of a
full-dimensional cell \(\mathcal{C}\) if and only if each pseudoline
properly intersects the boundary of \(\mathcal{C}\) exactly twice and their
intersections with its boundary alternate.
\end{observation}
This gives us a way to encode the location of the intersection of \(\ell_a\)
and \(\ell_b\) in \(\Xi\) using only bits. For an arrangement of pseudolines,
we use the standard vertical decomposition to construct a hierarchical cutting,
which guarantees that a pseudoline intersects a cell boundary at most twice.
For an arrangement of lines, we can use the standard bottom-vertex
triangulation instead, which will allow us to generalize our results to higher
dimensions in Section~\ref{sec:hyperplanes}. In the plane, the bottom-vertex
triangulation partitions the space into triangular cells. We define the
\emph{cyclic permutation} of a full-dimensional cell \(\mathcal{C}\) and a
finite set of pseudolines \(\mathcal{L}\) to be the finite sequence of properly
intersecting pseudolines from \(\mathcal{L}\) encountered when walking along
the boundary of \(\mathcal{C}\) in clockwise or counterclockwise order, up to
rotation and reversal.

Note that non-full-dimensional cells are easier to encode. For a 0-dimensional
cell and a pseudoline, we store whether the pseudoline lies above, lies below,
or contains the 0-dimensional cell. For a 1-dimensional cell, a pseudoline
could also intersect the interior of the cell, but in only one point. The
intersections with that cell define an (acyclic) permutation with potentially
several intersections at the same position. This information suffices to answer
location queries for those cells, and the space taken is not more than that
necessary for full-dimensional cells. When two pseudolines intersect in a
1-dimensional cell or contain the same 0-dimensional cell, they will
appear simultaneously in the cyclic permutation of an adjacent 2-dimensional
cell if they intersect its interior. If that is the case, the location of
the intersection of those two pseudolines in the cutting is the
non-full-dimensional cell. A constant number of bits can be added to the
encoding each time we need to know the dimension of the cell we encode.

\paragraph*{Encoding} Given \(n\) pseudolines in the plane and some fixed
parameter \(r\), compute a hierarchical \(\frac{1}{r}\)-cutting of those
pseudolines. This hierarchical cutting consists of \( \ell \) levels
labeled \(0,1,\ldots,\ell-1\). Level \(i\) has \(O(r^{2i})\)
cells. Each of those cells is further partitioned into
\(O(r^2)\) subcells. The \(O(r^{2(i+1)})\) subcells of level \(i\)
are the cells of level \(i+1\). Each cell of level \(i\) is intersected by at
most \(\frac{n}{r^i}\) pseudolines, and hence each subcell is intersected by at
most \(\frac{n}{r^{i+1}}\) pseudolines.

We compute and store a combinatorial representation of the hierarchical cutting
as follows: For each level of the hierarchy, for each cell in that level, for
each pseudoline intersecting that cell, for each subcell of that cell, we store
two bits to indicate the location of the pseudoline with respect to that
subcell, that is, whether the pseudoline lies above (\texttt{00}), lies below
(\texttt{01}) or intersects the interior of that subcell (\texttt{10}). When a
subcell is non-full-dimensional, we use another value (\texttt{11}) when the
pseudoline contains the subcell. When a pseudoline intersects the interior of a
2-dimensional subcell, we also store the two indices of the intersections of
that pseudoline with the subcell in the cyclic permutation associated with that
subcell, beginning at an arbitrary location in, say, clockwise order. If the
intersected subcell is 1-dimensional instead, we store the index of the
intersection in the acyclic permutation associated with that subcell, beginning
at an arbitrary endpoint. If two pseudolines intersect in the interior of a
1-dimensional subcell or on the boundary of a 2-dimensional subcell, they share
the same index in the associated permutation.

This representation takes \(O(\frac{n}{r^i} + \frac{n}{r^{i+1}}
\log{\frac{n}{r^{i+1}}})\) bits per subcell of level \(i\) by storing for each
pseudoline its location and, when needed, the permutation indices of its
intersections with the subcell. For each of the \(\lambda =
O(\frac{n^2}{t^2})\) subcells of the last level of the hierarchy we store a
pointer to a lookup table of size \(\tau = O(t^3)\) that allows to answer the
query of the orientation of any triple of pseudolines intersecting that
subcell. The number \(t = \frac{n}{r^\ell}\) denotes an upper bound on the
number of pseudolines intersecting each of those subcells.

Storing the permutation at each subcell would suffice to answer all
queries that do not reach the last level of the hierarchy. However, to get fast
queries, we need to have access to all bits belonging to a given pseudoline
without having to read the bits of the others. Using the Zone Theorem, and the
fact that hierarchical cuttings are constructed by decompositions of subsets of
the input pseudolines, we can bound the number of bits stored for a single
pseudoline intersecting a given cell of level \(i\) by \(\zeta_i = O(r^2 + r
\log{\frac{n}{r^{i+1}}})\). This allows us to store all bits belonging to a
given cell-pseudoline pair \((\mathcal{C} , \ell)\) in a contiguous block of
memory \(\sigma(\mathcal{C}, \ell)\) whose location in the encoding is easy to
compute. We call \(\sigma(\mathcal{C}, \ell)\) the \emph{signature} of \(\ell\)
in \(\mathcal{C}\). The overall number of bits stored stays the same up to a
constant factor.

For queries that reach the last level of the hierarchy, storing an individual
lookup table for each leaf would cost too much as soon as \(t = \omega(1)\).
However, as long as \(t\) is small enough, each order type is shared by
many leaves, and we can thus reuse space. Formally, let \(\nu(n)\) denote the
number of order types of size \(n\), which is \(\nu(n) = 2^{\Theta(n^2)}\) for
abstract order types and \(\nu(n) = 2^{\Theta(n \log{n})}\) for realizable
order types. At most \(\nu(t)\) distinct lookup tables are needed for answering
the queries on the subcells of the last level of the hierarchy. Hence the
pointers have size \(\Psi = O(\log \nu(t))\) and the total size needed for the lookup
tables is \(O(t^3 \nu(t))\). For each leaf, we store a canonical
labeling of size \(\kappa = O(t \log t)\) on the pseudolines that intersect it. We use
that canonical labeling to order the queries in the associated lookup table.

The encoding is the concatenation of the parameters \(n\), \(r\), and \( t\),
all signatures \(\sigma(\mathcal{C},\ell)\) for all pairs \((\mathcal{C},
\ell)\) with \(\ell \cap C\), the canonical labelings at the leaves, the leaf
pointers and the lookup tables. We define a canonical order on the cells of
level \(i\) so that they can be ordered as \(\{\, \mathcal{C}_{i,1},
\mathcal{C}_{i,2}, \ldots, \mathcal{C}_{i,O(r^{2i})}\,\}\).
We complement the encoding with appropriate padding so that the
position \(\rho(\mathcal{C}_{i,j}, \ell)\) of the signature
\(\sigma(\mathcal{C}_{i,j}, \ell)\) is
\begin{displaymath}
\rho(\mathcal{C}_{i,j}, \ell)
= \rho(\mathcal{C}_{i-1,1}, \ell_0^{\mathcal{C}_{i-1,1}})
+ c r^{2i-2} \left\lfloor \frac{n}{r^{i-1}} \right\rfloor \zeta_{i-1}
+ (j-1) \left\lfloor \frac{n}{r^i} \right\rfloor \zeta_i
+ \pi(\mathcal{C}_i, \ell) \zeta_i,
\end{displaymath}
where \(c\) is some constant, \(\pi(\mathcal{C}, \ell)\)
is the first index of \(\ell\) in the cyclic permutation of \(\mathcal{C}\),
\(\ell_a^{\mathcal{C}}\) is the pseudoline such that \(\pi(\mathcal{C},
\ell_a^{\mathcal{C}}) = a\), and, for the root cell of the hierarchy
\(\mathcal{C}_{0,1}\) representing the entire space and containing all the
intersections of the arrangement, \(\rho(\mathcal{C}_{0,1}, \ell_0)\) is the
position after the encoding of the parameters \(n\), \(r\), and \(t\), and
\(
\rho(\mathcal{C}_{0,1}, \ell_a) = \rho(\mathcal{C}_{0,1}, \ell_0) + a \zeta_0
\).
The canonical labeling of the first leaf is stored at position
\begin{displaymath}
\rho_{\Lambda_0} = c \sum_{i=0}^{\ell-1} r^{2i} \left\lfloor \frac{n}{r^i}
\right\rfloor \zeta_i,
\end{displaymath}
the lookup table pointer of the first leaf is stored at position
\(\rho_{\Lambda_0} + \kappa\), the canonical labeling of leaf \(\Lambda\) is stored
at position \(\rho_{\Lambda_0} + ( \Lambda - 1 ) ( \kappa + \Psi ) \) and its
table lookup pointer at \(\kappa\) from that position. The first lookup table
is stored at position \(\rho_{\Lambda_0} + \lambda ( \kappa + \Psi )\) and
lookup table \(\Theta\) is stored at position
\(\rho_{\Lambda_0} + \lambda ( \kappa + \Psi ) + (\Theta - 1) \tau\).

\paragraph*{Space Complexity}
We prove a general bound on the space taken by our construction when the
hierarchy contains \(\ell\) levels.
Let $H_r^\ell(n) \in \mathbb{N}$ be the maximum amount of space (bits), over
all arrangements of \(n\) pseudolines, taken by the $\ell \in \mathbb{N}$
levels of a hierarchy with parameter $r \in (1,+\infty)$. This excludes the
space taken by the lookup tables, their associated pointers and canonical
labelings at the leaves, and the parameters of the hierarchy \(n\), \(r\) and
\(t\).

\begin{lemma}\label{lem:space-2-hierarchy}
For \( r \geq 2 \) we have
\begin{displaymath}
H_r^\ell(n)
=
O\left(\frac{n^2}{t} ( \log t + r )\right).
\end{displaymath}
\end{lemma}

\begin{proof}
By definition, we have
\begin{displaymath}
H_r^\ell(n)
= O \left(
\sum_{i=0}^{\ell-1} \left(
	r^{2i} \cdot r^2 \cdot \left(
\frac{n}{r^i} + \frac{n}{r^{i+1}} \log \frac{n}{r^{i+1}}
\right)
\right)
\right).
\end{displaymath}
We multiply the previous equation by \(\frac{n}{tr^\ell} = 1\)
\begin{displaymath}
H_r^\ell(n)
= O \left(
\frac{n^2}{t}
\sum_{i=0}^{\ell-1} \left(
\frac{1}{r^{\ell - i - 1}} \cdot \left(
	r + \log \frac{n}{r^{i+1}}
\right)
\right)
\right).
\end{displaymath}
We use the equivalence \(\frac{n}{r^{i+1}} = t r^{\ell - i - 1}\) to replace
the last term in the previous equation
\begin{displaymath}
H_r^\ell(n)
= O \left(
\frac{n^2}{t}
\sum_{i=0}^{\ell-1} \left(
\frac{1}{r^{\ell - i - 1}} \cdot \left(
	r + \log t + (\ell - i - 1) \log r
\right)
\right)
\right).
\end{displaymath}
We reverse the summation by redefining \(i \leftarrow \ell - i - 1\) and group the terms
\begin{displaymath}
H_r^\ell(n)
=
O \left(
\frac{n^2}{t} \left(
	(\log t + r)  \sum_{i=0}^{\ell-1} \frac{1}{r^{i}}
	+
	\log r \sum_{i=0}^{\ell-1} \frac{i}{r^{i}}
\right)
\right).
\end{displaymath}
Using the following inequalities:
\begin{displaymath}
\sum_{i = 0}^{k} x^i
\le
\frac{1}{1-x}
\qquad
\text{and}
\qquad
\sum_{i = 0}^{k} i x^i
\le
\frac{x}{{(1-x)}^2},
\qquad
\forall k \in \mathbb{N},
\forall x \in (0,1),
\end{displaymath}
we conclude that
\begin{displaymath}
H_r^\ell(n)
=
O\left(
\frac{n^2}{t} \left(
	\left(1 + \frac{1}{r-1} \right) (\log t + r)
	+
	\left(1 + \frac{2r-1}{r^2-2r+1}\right) \frac{\log r}{r}
\right)
\right),
\end{displaymath}
and that for $r \geq 2$
\begin{displaymath}
H_r^\ell(n)
=
O\left(\frac{n^2}{t} (\log t + r)\right). \qedhere
\end{displaymath}
\end{proof}

Taking into account the space taken by the other bits of the encoding we
obtain
\begin{lemma}\label{lem:space-2-all}
The space taken by our encoding is
\begin{displaymath}
	S_r^\ell(n) =
	O\left(
	\log ntr +
	\frac{n^2}{t} ( \log t + r ) +
	t^3 \nu(t) + \frac{n^2}{t^2} ( \log \nu(t) + t \log t )
	\right).
\end{displaymath}
\end{lemma}

We pick \(r\) constant for both abstract and realizable order type. For
abstract order types, we choose \(t = \sqrt{\delta \log n}\) for small enough
\(\delta\) and the last term in Lemma~\ref{lem:space-2-all} dominates with \(n^2\).
Note how the quadratic bottleneck of this encoding is the storage of the order
type pointers at the leaves of the hierarchy.
For realizable order types, we choose \(t = \delta \log n / \log \log n\) for
small enough \(\delta\) and the second and last term in
Lemma~\ref{lem:space-2-all} dominate with \(n^2 {(\log \log n)}^2 / \log n\).
This proves the space constraints in
Theorems~\ref{thm:abstract}~and~\ref{thm:realizable}.

\paragraph*{Correctness and Query Complexity} Given our encoding and three
pseudoline indices \(a,b,c\) we answer a query as follows: We start by decoding the
parameters \(n\), \(r\), and \(t\). In the word-RAM model, this can be done in
\(O(\log^* n + \log^* r + \log^* t)\) time.\footnote{%
Encode \(n\) in binary using \(\lceil \log n \rceil\) bits, \(\lceil \log n \rceil\) using \(\lceil
\log \log n \rceil\) bits, \(\lceil \log \log n \rceil\) using \(\lceil \log \log \log
n \rceil\) bits, etc.\ until the number to encode is smaller than a constant
which we encode in unary with \texttt{1}'s. Prepend a \texttt{1} to the
largest number and \texttt{0} to all the others except the smallest.
Concatenate those numbers from smallest to largest. Total space is \(O(\log
n)\) bits and decoding \(n\) can be done in \(O(\log^* n)\) time in the
word-RAM model with \(w \geq \log n\). Logarithmic space and constant
decoding time is more easily attained when \(w = \Theta(\log n)\).%
}
Let \(\mathcal{C} = \mathcal{C}_{0,1}\).
First, find the subcell \(\mathcal{C}'\) of \(\mathcal{C}\) containing \(\ell_a
\cap \ell_b\) by testing for each subcell whether the intersections of
\(\ell_a\) and \(\ell_b\) with the subcell alternate in the cyclic permutation.
This can be done in \(O(r^2)\) time by scanning \(\sigma(\mathcal{C},
\ell_a)\) and \(\sigma(\mathcal{C}, \ell_b)\) in parallel. Note that non-full
dimensional subcells can be tested more easily. Next, if \(\ell_c\) does
not properly intersect \(\mathcal{C}'\), answer the query accordingly. If on
the other hand \(\ell_c\) does properly intersect the subcell we recurse on
\(\mathcal{C}'\). This can be tested by scanning \(\sigma(\mathcal{C},
\ell_c)\) in \(O(r^2)\) time. Note that in case that the subcell is
non-full-dimensional we can already answer the query at this point. When we
reach the relative interior of a subcell of the last level of the hierarchy
without having found a satisfactory answer, we can answer the query by table
lookup in constant time. This works as long as each order type identifier for
at most \(t\) pseudolines fits in a constant number of words, which is the case
for the values of \(t\) we defined.
The position of the signatures scanned during the first recursive step of the
query can be computed in constant time and at each other recursive step of
the query we can compute the positions of the signatures we need to scan
from the position of the signatures scanned during the previous recursive step
in constant time.
When we reach the bottom of the recursion, the position of the lookup table
pointer, the position of the canonical labeling, and the position of the lookup
table can be computed in constant time.
The total query time is thus proportional
to \(r^2 \log_r n\) in the worst case, which is logarithmic since \(r\) is
constant.
This proves the query time constraints in
Theorems~\ref{thm:abstract}~and~\ref{thm:realizable}.
With the hope of getting faster queries we could pick \(r = \Theta(\log
t)\) to reduce the depth of the hierarchy, without changing the
space requirements by more than a constant factor.
However, if no additional care is taken, this would slow the queries down by a
\(\Theta(\log^2 t / \log \log t)\) factor because of the scanning approach
taken when locating the intersection \(\ell_a \cap \ell_b\). We show how to handle
small but superconstant \(r\) properly in the next section.

\paragraph*{Preprocessing Time}
For a set of \(n\) points in the plane, or an arrangement of \(n\) lines in the
dual, we can construct the encoding of their order type in quadratic time in
the real-RAM and constant-degree algebraic computation tree models. We prove
Theorem~\ref{thm:preprocessing}.
\begin{proof}
A hierarchical cutting can be computed in \(O(nr^\ell)\) time in the dual
plane. All signatures \(\sigma(\mathcal{C}, \ell)\) can be computed from the cutting
in the same time. The lookup tables and leaf-table pointers can be computed
in \(O(n^2 + t^3 \nu(t))\) time as follows: For each subcell \(\mathcal{C}\)
of the \(\frac{n^2}{t^2}\) subcells of the last level of the hierarchy,
compute a canonical labeling and representation of the lines intersecting
\(\mathcal{C}\) in \(O(n^2)\) time as in~\cite{AILOW14}.
Insert the canonical representation in some trie in \(O(t^2)\)
time. If the canonical representation was not already in the trie, create a lookup
table with the answers to all \(O(t^3)\) queries on those lines and attach a
pointer to that table in the trie. This happens at most \(\nu(t)\) times.
In the encoding, store the canonical labeling and this new pointer or the
pointer that was already in the trie for the subcell \(\mathcal{C}\). All
parts of the encoding can be concatenated together in time proportional to
the size of the encoding.
\end{proof}
\section{Sublogarithmic Query Complexity}\label{sec:query}

We further refine the data structure defined in the previous section so as to
reduce the query time by a \( \log{\log{n}} \) factor. We do so using
specificities of the word-RAM model that allow us to preprocess computations
on inputs of small but superconstant size. The idea is to make each
signature \(\sigma(\mathcal{C}, \ell)\) fit in a single word of memory by only
approximately encoding the cyclic permutation of the intersections around each
subcell, relying on the fact that ambiguous situations rarely arise. Those
ambiguous situations, if they happen, can be deterministically handled using
additional lookup tables.
This improvement is applicable for both abstract and realizable order types.

We improve our main theorems for the two-dimensional case:
\begin{theorem}\label{thm:abstract-loglog}
All abstract order types have a
\((O(n^2), O(\frac{\log{n}}{\log{\log{n}}}))\)-encoding.
\end{theorem}
\begin{theorem}\label{thm:realizable-loglog}
All realizable order types have a
\((O(\frac{n^2\log^\varepsilon n}{\log n}),
O(\frac{\log{n}}{\log{\log{n}}}))\)-encoding.
\end{theorem}
\begin{theorem}\label{thm:preprocessing-loglog}
In the real-RAM model and the constant-degree algebraic decision tree model,
given \(n\) real-coordinate input points in \(\mathbb{R}^2\) we can compute
the encoding of their order type as in
Theorems~\ref{thm:abstract-loglog}~and~\ref{thm:realizable-loglog} in
\(O(n^2)\) time.
\end{theorem}

\paragraph*{Bit Packing}
Fix a large \(\alpha\), a small \(\delta\) and define \(r = \Theta(\log^\delta
n)\). Note that we can construct a hierarchical cutting with superconstant
\(r\) by constructing a hierarchical cutting with some appropriate constant
parameter \(r'\), and then skip levels that we do not need. Denote by \(n_i =
n/r^i\) an upper bound on the number of lines intersecting a cell of level
\(i\). For each subcell of level \(i\), partition its cyclic permutation into
\(\log^\alpha{n}\) blocks of at most \(n_{i+1} / \log^\alpha n \leq n /
{(\log n)}^{\alpha + \delta (i + 1)}\) intersections.
For each pseudoline intersecting a cell we only store the block numbers that
that pseudoline touches in its signature. Note that now each signature
\(\sigma(\mathcal{C}, \ell)\) only uses \(\Theta (\log^{2\delta}{n} + \alpha
\log^\delta{n} \log{\log{n}}) = \Theta(\log^{2\delta}{n})\) bits, which fits in
a word for small enough \(\delta\).

\paragraph*{Intersection Oracle}
We construct an additional lookup table to compute the subcell in which
\(q_i \cap q_j\) lies in constant time. Computing it via scanning with so many
subcells to check would waste any further savings. For that we need
a general observation on the precomputation of functions on small universes.
\begin{observation}
In the word-RAM model, for any word-to-word function $f:[2^w] \to [2^w]$,
we can build a lookup table of total bitsize $2^{s+1} w$
for all $2^s$ inputs $x \in [2^s]$ of bitsize $s \le w$ in time $2^s T(s)$
where $T(s)$ is the complexity of computing $f(x)$, $x \in [2^s]$.
The image of any input of bitsize $s$ can then be retrieved in $O(1)$ time by a
single lookup (since the input fits in a single word). In particular,
we have $2^s T(s)$ and
\(2^{s+1} w\) sublinear as long as \(T(s)=s^{O(1)}\) and $s \le (1-\varepsilon) \log_2 n$.
\end{observation}
In other words, any polynomial time computable word-to-word function can be
precomputed in sublinear time and space for all inputs of roughly logarithmic
size.

Since our pseudoline identifiers now fit in \(\Theta (\log^{2\delta}{n})\) bits
we can choose an appropriate \(\delta\) so as to satisfy the requirements given
above. This means we can precompute the function that sends two pseudoline
identifiers to either the subcell containing their intersection or to some
special value in case of an ambiguous input.

\paragraph*{Disambiguation}
Note that ambiguous inputs rarely occur. An input is ambiguous if and only if
at least one boundary intersection of each pseudoline appears in the same cyclic
permutation block of the cell that contains their intersection. This happens
less than
\(\log^\alpha n \cdot {({n_{i+1}} / \log^\alpha n)}^2 =
\log^\alpha n \cdot \frac{n^2}{r^{2(i+1)}} / \log^{2\alpha} n =
\frac{n^2}{r^{2(i+1)}} / \log^\alpha n\) times per subcell of level \(i\) of
the hierarchy. Summing over all levels, we get a subquadratic size lookup table
for ambiguous cases which can be implemented using standard tools.

\paragraph*{Space Complexity}
The space used by the hierarchy is proportional to
\begin{displaymath}
\sum_{i=0}^{\ell-1} r^{2i} \cdot r^2 \cdot \left( \frac{n}{r^i} +
\frac{n}{r^{i+1}} \alpha \log \log n \right) =
O\left(\frac{n^2}{t} ( \log^\delta n + \alpha \log \log n )\right).
\end{displaymath}
Intersection oracles and disambiguation tables fit in subquadratic space and the
space analysis for the rest of the data structure still holds.
For small enough \(\delta\) the space remains quadratic for abstract order types
and subquadratic for realizable order types.
This proves the space constraints in
Theorems~\ref{thm:abstract-loglog}~and~\ref{thm:realizable-loglog}.
Unfortunately, we must incur a nonabsorbable extra \(\log^\delta n\) factor
in the realizable case. Note that a \(\log{\log{\log{n}}}\) factor can be
squeezed without increasing the space usage by choosing \(r = \Theta(\log \log
n)\)
instead.

\paragraph*{Correctness and Query Complexity} The previous analysis still holds
modulo additional disambiguation lookups and oracle-based intersection
location. We now have a shallower decision tree of depth \(\log_r{n} =
O_\delta(\frac{\log{n}}{\log{\log{n}}})\). This proves the query time
constraints in
Theorems~\ref{thm:abstract-loglog}~and~\ref{thm:realizable-loglog}.

\paragraph*{Preprocessing Time}
We prove Theorem~\ref{thm:preprocessing-loglog}.
\begin{proof}
As before, the hierarchical cutting and all signatures \(\sigma(C, \ell)\)
can be computed in \(O(nr^\ell)\) time. The lookup table and leaf-table
pointers can be computed in \(O(n^2)\) time. All intersection oracles and
disambiguation tables can be computed in subquadratic time.
\end{proof}
\section{Higher-Dimensional Encodings}\label{sec:hyperplanes}

We can generalize our point configuration encoding to any dimension \(d\). The
chirotope of a point set in \(\mathbb{R}^d\) consists of all orientations of
simplices defined by \(d+1\) points of the set. The orientation of the simplex
with \(d+1\) ordered vertices \(p_i\) with coordinates \((p_{i,1} , p_{i,2} ,
\ldots, p_{i,d} )\) is given by the sign of the determinant
\begin{displaymath}
\begin{vmatrix}
1 & p_{1,1} & p_{1,2} & \hdots & p_{1,d} \\
1 & p_{2,1} & p_{2,2} & \hdots & p_{2,d} \\
\vdots & \vdots & \vdots & \ddots & \vdots \\
1 & p_{d+1,1} & p_{d+1,2} & \hdots & p_{d+1,d}
\end{vmatrix}.
\end{displaymath}

We obtain the following generalized result:
\begin{theorem}\label{thm:realizable-d}
All realizable chirotopes of rank \(k \geq 4\) have a
\((O(\frac{n^{k-1}{(\log{\log{n}})}^2}{\log n}),
O(\frac{\log{n}}{\log{\log{n}}}))\)-encoding.
\end{theorem}
\begin{theorem}\label{thm:preprocessing-d}
In the real-RAM model and the constant-degree algebraic decision tree model,
given \(n\) real-coordinate input points in \(\mathbb{R}^d\) we can compute
the encoding of their chirotope as in Theorem~\ref{thm:realizable-d} in \(
O(n^{d}) \) time.
\end{theorem}

\paragraph*{Intersection Location}
In the primal, the orientation of a simplex can be interpreted as the location
of its last vertex with respect to the (oriented) hyperplane spanned by its
base. In the dual, this orientation corresponds to the location of the
intersection of the \(d\) first dual hyperplanes with respect to the last
(oriented) dual hyperplane.
We gave an encoding for the two-dimensional case in Section~\ref{sec:abstract}.
With this encoding, a query is answered by traversing the levels of
some hierarchical cutting, branching on the location of the intersection of two
of the three query lines.  We generalize this idea to \(d\) dimensions. Now the
cell considered at the next level of the hierarchy will depend on the location
of the intersection of \(d\) of the \(d+1\) query hyperplanes.

In two dimensions, we solved the following subproblem:

\begin{problem}
Given a triangle and \(n\) lines in the plane, build a data structure that,
given two of those lines, allows to decide whether their intersection
intersects the triangle.
\end{problem}

In retrospective, we showed that there exists such a data structure using
\(O(n \log{n})\) bits that allows for queries in \(O(1)\) time. We generalize
this result. Consider the following generalization of the problem in
\(d\) dimensions:

\begin{problem}
Given a convex body and \(n\) hyperplanes in \(\mathbb{R}^d\), build a data
structure that, given \(d\) of those hyperplanes, allows to decide whether their
intersection intersects the convex body.
\end{problem}

Of course this problem can be solved using \(O(n^d)\) space by explicitly
storing the answers to all possible queries. This is best possible for
\(d=1\). For \(d \geq 2\), we show how to reduce the space to \(O(n^{d-1} \log
n)\) by recursing on the dimension, taking \(d = 2\) as the base case.

We encode the function that maps a \(d\)-tuple of indices of input dual
hyperplanes \(H_i\) to \(1\) if their intersection intersects a fixed convex
body \(C\) and \(0\) otherwise.
\begin{displaymath}
\mathcal{I} \colon\, {[n]}^d \to \{\, 0,1\,\} \colon\,
(i_1,i_2,\ldots,i_d) \mapsto
(H_{i_1} \cap H_{i_2} \cap H_{i_3} \cap H_{i_4} \cap \cdots \cap
H_{i_d})
\cap C
\neq \emptyset.
\end{displaymath}
We call this function the \emph{intersection function} of \((C,H)\).
We prove that
\begin{theorem}
All intersection functions have a \((O(n^{d-1} \log n),O(d))\)-encoding.
\end{theorem}

\begin{proof}
Consider a convex body \(C\) and \(n\) hyperplanes \(H_i\). We want a data
structure that can answer any query of the type
\begin{displaymath}
(H_{i_1} \cap H_{i_2} \cap H_{i_3} \cap H_{i_4} \cap \cdots \cap
H_{i_{d-1}} \cap H_{i_d})
\cap C
\neq \emptyset.
\end{displaymath}

Note that this is equivalent to deciding whether
\begin{displaymath}
(H_{i_1} \cap H_{i_2} \cap H_{i_3} \cap H_{i_4} \cap \cdots \cap H_{i_{d-1}})
\cap (H_{i_d} \cap C)
\neq \emptyset,
\end{displaymath}
where \(H_{i_d} \cap C\) is a convex body of dimension \(d-1\), and the
number of hyperplanes we want to intersect it with is \(d-1\).

We unroll the recursion until the convex body is of dimension two, and only
two hyperplanes are left to intersect. We then notice that the decision we
are left with is equivalent to
\begin{displaymath}
(H_{i_1} \cap (H_{i_3} \cap H_{i_4} \cap \cdots \cap H_{i_d}))
\cap
(H_{i_2} \cap (H_{i_3} \cap H_{i_4} \cap \cdots \cap H_{i_d}))
\cap (C \cap (H_{i_3} \cap H_{i_4} \cap \cdots \cap H_{i_d}))
\neq \emptyset,
\end{displaymath}
which reads: ``Given two lines and a convex body in some plane, do they
intersect?''.
We can answer this query if we have the encoding for \(d = 2\) which is
obtained by replacing \emph{triangle} by \emph{convex body} in the
two-dimensional original problem. The total space taken is proportional to
the number of times we unroll the recursion times the space taken in two
dimensions, which is \(n^{d-2} \cdot n \log{n} = O(n^{d-1} \log{n}) \).
Queries can then be answered in \(O(d)\) time.

Note that degenerate cases can arise: convex bodies that are empty because of
nonintersecting hyperplanes and convex bodies that are higher dimensional because
of linearly dependent hyperplanes. However, all those cases can be
dealt with appropriately. We omit the details here.
\end{proof}

In the paragraphs that follow, we show how to plug this result in those of
the previous sections to obtain analogous results for the \(d\)-dimensional
version of the problem.

\paragraph*{Space Complexity}
We first generalize Lemma~\ref{lem:space-2-hierarchy} of
Section~\ref{sec:abstract}. Let $H_r^\ell(n,d) \in \mathbb{N}$ be the maximum
amount of space (bits), over all arrangements of \(n\) hyperplanes in \(\mathbb{R}^d\),
taken by the $\ell \in \mathbb{N}$ levels of a hierarchy with parameter $r \in
(1,+\infty)$. As before, this excludes the space taken by the lookup tables
and associated pointers at the leaves.

\begin{lemma}\label{lem:space-d-hierarchy}
For \( r \geq 2 \) we have
\begin{displaymath}
H_r^\ell(n,d)
=
O\left(\frac{n^d}{t} (\log t + \frac{r}{t^{d-2}})\right).
\end{displaymath}
\end{lemma}

\begin{proof}
By definition we have
\begin{displaymath}
H_r^\ell(n,d)
= O \left(
\sum_{i=0}^{\ell-1} \left(
	r^{di} \cdot r^d \cdot \left(
	  \frac{n}{r^i} + {\left(\frac{n}{r^{i+1}}\right)}^{d-1} \cdot \log \frac{n}{r^{i+1}}
\right)
\right)
\right).
\end{displaymath}
Using \(t = n/r^\ell\), reversing the summation with
\(i \leftarrow \ell - i - 1\), and grouping the terms, we have
\begin{displaymath}
H_r^\ell(n,d)
=
O \left(
\frac{n^d}{t} \left(
	\frac{r}{t^{d-2}}
	\sum_{i=0}^{\ell-1} \frac{1}{{(r^{d-1})}^i}
	+
	\log t
	\sum_{i=0}^{\ell-1} \frac{1}{r^i}
	+
	\log r
	\sum_{i=0}^{\ell-1} \frac{i}{r^i}
\right)
\right).
\end{displaymath}
Using the same inequalities as in the proof of Lemma~\ref{lem:space-2-hierarchy}
we conclude that for $r \geq 2$
\begin{displaymath}
H_r^\ell(n,d)
=
O\left(\frac{n^d}{t} \left(\log t + \frac{r}{t^{d-2}}\right)\right). \qedhere
\end{displaymath}
\end{proof}

Taking into account the space taken by the other bits of the encoding
we obtain
\begin{lemma}\label{lem:space-d-all}
The space taken by our data structure is
\begin{displaymath}
	S_r^\ell(n,d) =
	O\left(
	  \log dntr +
	  \frac{n^d}{t} ( \log t + \frac{r}{t^{d-2}} ) + t^{d+1} \nu(t,d+1) +
	  \frac{n^d}{t^d} ( \log \nu(t,d+1) + t \log t )
	\right),
\end{displaymath}
where \(\nu(n,k) = 2^{\Theta({(k-1)}^2 n \log{n})}\) denotes the number of
realizable rank-\(k\) chirotopes of size \(n\).
\end{lemma}

We pick \(r\) constant and choose \(t = \delta \log n / \log \log n\) for small
enough \(\delta\). The second term in Lemma~\ref{lem:space-2-all} dominates
with \(n^d {(\log \log n)}^2 / \log n\). This proves the space constraint in
Theorem~\ref{thm:realizable-d}.

\paragraph*{Correctness and Query Complexity}
As before, a query is answered by traversing the hierarchy, which takes
\(O(\log n)\) time. The query time can be further improved using the method from
Section~\ref{sec:query} with \(r = \Theta(t^{d-2} \log t)\). This proves the
query time constraint in Theorem~\ref{thm:realizable-d}.

\paragraph*{Preprocessing Time}
We prove Theorem~\ref{thm:preprocessing-d}.
\begin{proof}
The hierarchical cuttings can be computed in \(O(n{(r^\ell)}^{d-1})\) time.
The lookup table and leaf-table pointers can be computed in \(O(n^d)\) time
using the canonical labeling and representation for rank-\((d+1)\) chirotopes
given in~\cite{AILOW14}. All intersection oracles and disambiguation tables
can be computed in \(o(n^{d})\) time.
\end{proof}
\section{Conclusion}

Observe the following. Assume we are given an instance of some real-input
decision problem. Given a decision tree of depth \(D\) for this problem for which
each input query and answer can be encoded using at most \(Q\) bits, we can
encode the instance using at most \(DQ\) bits by encoding the path traversed when
executing the decision tree on this instance. The general position testing
problem (GPT) asks if an input set of \(n\) points in the plane contains a
collinear triple. It is an example of a real-input decision problem for which
no subquadratic real-RAM algorithm is known even though the best known lower
bound is only linearithmic. Since shallow decision trees yield short encodings,
we see the design of a subquadratic encoding for realizable order types as a
stepping stone towards nonuniform and uniform subquadratic algorithms for
GPT\@.

Unfortunately, even though our encodings achieve subquadratic space for
realizable order types, they cannot be used to test for isomorphism in
subquadratic time. This is partly because the preprocessing time to construct
the encoding is already quadratic.
However, observe that the preprocessing time we achieve in this contribution
matches the best known upper bound for GPT in the algebraic decision tree
model:
\begin{theorem}
If there is an encoding with construction cost \(C(N)\) for
realizable order types in the algebraic decision tree model, then
there is a nonuniform algorithm for general position testing that runs
in time \(C(N)\) in the algebraic decision tree model.
\end{theorem}
\begin{proof}
Construct the encoding. Then at zero cost in the nonuniform
model, run all \(O(n^3)\) queries on the encoding.
\end{proof}
Goodman and Pollack saw GPT as a multidimensional generalization of
sorting~\cite{GP83}.
We again stress the need for a better understanding of this fundamental
problem.
\bibliography{paper}

\end{document}